\documentclass[journal]{IEEEtran}
\usepackage{amsmath,amsfonts}
\usepackage{algorithmicx,algorithm,algpseudocode}
\usepackage{array}
\usepackage[caption=false,font=normalsize,labelfont=sf,textfont=sf]{subfig}
\usepackage{textcomp}
\usepackage{stfloats}
\usepackage{url}
\usepackage{verbatim}
\usepackage{graphicx}
\usepackage{cite}
\usepackage{svg}
\newtheorem{theorem}{Theorem}
\usepackage{multirow}
\usepackage{booktabs}

\begin{document}

\title{LLM-Driven Large-Scale Spectrum Access}

\author{{Ning Yang$^*$,~\IEEEmembership{Senior Member,~IEEE}, Jinliang Gao, Haijun Zhang,~\IEEEmembership{Fellow,~IEEE}}
\thanks{Ning Yang is with the Institute of Automation, Chinese Academy of Sciences, Beijing, 100190, China (e-mail: ning.yang@ia.ac.cn).}
\thanks{Jinliang Gao is with the College of Computer Science, Sichuan University, Chengdu, Sichuan, 610065, China (e-mail: gaojinliang666@gmail.com).}
\thanks{Haijun Zhang is with the School of Computer and Communication Engineering, University of Science and Technology Beijing, Beijing, 100083, China (e-mail: zhanghaijun@ustb.edu.cn).}
\thanks{This work was supported  by the National Natural Science Foundation of China under Grants 62301559.}
\thanks{($^*$Corresponding Author: Ning Yang)}
}



\maketitle

\begin{abstract} 
Efficient spectrum management in massive-scale wireless networks is increasingly challenged by explosive action spaces and the computational intractability of traditional optimization. This study proposes a LLM-Driven Large-Scale Spectrum Access (LSA) framework rooted in Group Relative Policy Optimization (GRPO). To overcome the computational intractability caused by ultra-long prompts in large-scale scenarios, we develop a hierarchical state serialization mechanism that synthesizes global environment statistics with localized critical constraints, enabling the LLM to perform high-dimensional reasoning within a bounded context window. Simulation results under strictly time-bounded inference protocols reveal that the code-driven paradigm eliminates the Supervised Fine-Tuning (SFT) cold-start bottleneck and leverages direct execution feedback to achieve superior scaling laws. The framework maintains robust spectral utility and generalization across varying network scales, yielding consistent and empirically superior performance over stochastic heuristics, and surpassing partitioned classical solvers in ultra-dense regimes under matched compute budgets. Code is available at https://github.com/Xtdzs/LLM-Driven-Large-Scale-Spectrum-Access.
\end{abstract}


\section{Introduction}

\IEEEPARstart{T}{he} rapid evolution of 6G wireless networks has catalyzed an unprecedented demand for high-throughput transmission and massive connectivity. To support the burgeoning requirements of the Internet of Things (IoT) and autonomous systems, Intelligent Spectrum Access (ISA) has emerged as a fundamental paradigm shift from static frequency allocation to autonomous, dynamic management. By allowing heterogeneous nodes to opportunistically access underutilized bands, ISA significantly enhances spectral efficiency in dense environments \cite{ref1, ref2}.

However, despite its theoretical potential, spectrum management in large-scale scenarios faces severe technical bottlenecks that hinder practical deployment. Modern wireless environments are increasingly defined by high dimensionality, extreme dynamism, and complex coupling between network parameters. Consequently, as node density increases, the joint optimization of channel selection and interference alignment becomes computationally intensive, frequently reaching NP-hard complexity in decentralized settings. Furthermore, in massive deployments, the signaling overhead required for centralized multi-agent coordination often scales exponentially, eventually exceeding the available bandwidth. This will create a critical bottleneck where the latency of coordination exceeds the coherence time of the wireless channel, leading to diminished efficiency gains or even complete system synchronization failure.

Existing solutions, primarily based on Deep Reinforcement Learning (DRL), exhibit notable limitations when applied to these high-dimensional and non-stationary settings \cite{ref3, ref4}. While DRL improves the handling of larger state spaces compared to traditional mathematical optimization, standard agents often suffer from low sample efficiency and slow convergence when facing explosive action spaces. A primary cause of this instability is the moving target problem inherent in multi-agent environments; as all agents learn and update their policies simultaneously, the environment becomes non-stationary from the perspective of any individual agent, leading to severe policy oscillation and divergent training curves. In addition to these convergence issues, conventional RL models frequently rely on rigid, hand-crafted reward structures that fail to balance individual throughput with global fairness in fluctuating channel conditions. Most critically, existing models often lack the inherent generalization power required to adapt to unseen network topologies\cite{ref5} without extensive retraining, which is computationally prohibitive for resource-constrained edge devices.

To address these multi-faceted challenges, we propose a Large Language Model (LLM)-driven Spectrum Access (LSA) framework under high-dimensional and combinatorial constraints, where classical optimization and deep reinforcement learning methods suffer from poor scalability and limited generalization. To address this, we reformulate spectrum allocation as a structured reasoning problem and leverage a code-specialized LLM trained with group-relative policy optimization (GRPO), enabling scalable policy generation with deterministic feasibility repair.

The primary contributions of this work are summarized as follows:

\begin{itemize}
\item \textbf{Constraint-Aware Semantic Reasoning Architecture:} A novel framework is developed that maps combinatorial spectral constraints into the LLM's native reasoning manifold. By employing hierarchical state serialization and structural prompt injection, the architecture transforms high-dimensional physical observations into logically executable semantic directives. This distilled representation filters stochastic environmental noise while preserving critical interference dependencies, effectively overcoming the information redundancy and brittle generalization that hinder conventional recurrent or MLP-based agents.
\item \textbf{Scalable Policy Alignment via GRPO:} A scalable optimization methodology is proposed by integrating GRPO with a hierarchical state serialization mechanism. By employing group-relative baselines to stabilize policy gradients, this approach aligns the LLM's decision logic with system-level throughput and collision mitigation objectives. This integrated design facilitates stable policy convergence in massive-scale regimes while eliminating the computational overhead typically associated with approximating separate value-function critics.
\item \textbf{Paradigm Validation and Large-Scale Scaling:} A rigorous dual-path evaluation demonstrates that code-specialized backbones trained with pure GRPO inherently capture structural and algorithmic priors, bypassing the need for supervised cold-start. Extensive simulations confirm that this reasoning-centric framework achieves competitive performance in small-scale regimes and establishes a favorable scaling law, surpassing classical partitioned solvers in ultra-dense configurations of up to 2000 nodes under equivalent runtime constraints.
\end{itemize}

The remainder of this paper is organized as follows. Section II reviews related work in ISA and LLM applications. Section III defines the mathematical system model and the spectrum access problem. Section IV details the technical components of the LSA framework and the GRPO methodology. Section V outlines the experimental setup, including simulation parameters, baseline configurations, and evaluation protocols. Section VI provides a comprehensive analysis of the experimental results, and Section VII concludes the paper with a discussion.

\section{Related Work}

The development of ISA has transitioned from model-based optimization to data-driven autonomous learning. This section categorizes existing literature into three primary research directions.

\subsection{Conventional Mathematical and Game-Theoretic Frameworks}
Early research on Dynamic Spectrum Access (DSA) focused on rigorous mathematical modeling and optimization. To characterize the stochastic nature of primary user behaviors, renewal-theoretical approaches \cite{ref6} and queuing analysis \cite{ref7} were introduced to model temporal channel availability. In multi-user environments, game theory has been widely used to model competitive or cooperative interactions. Contract-theoretic models \cite{ref1, ref8} and auction mechanisms \cite{ref9} were developed to incentivize spectrum sharing. Furthermore, matching theory \cite{ref10, ref11} and graph coloring formulations \cite{ref12} have been applied to coordinate allocation and mitigate interference.

These conventional methods are heavily model-dependent, requiring precise prior knowledge of channel statistics that are often unavailable in real-world 6G deployments. Additionally, the computational complexity of solving these optimization problems, such as Lyapunov optimization \cite{ref13} or aggregative games \cite{ref14}, scales exponentially with the node count, rendering them intractable for large-scale IoT networks.

\subsection{Reinforcement Learning-Based Autonomous Access}
To address model dependency, reinforcement learning (RL) enables agents to learn policies through environmental interaction. Initial frameworks utilized Markov Decision Processes (MDP) to establish sensing sequences \cite{ref15, ref16}, while subsequent efforts employed multi-armed bandit (MAB) formulations \cite{ref11, ref20} and stochastic learning automata \cite{ref17, ref18}. The integration of deep learning led to DRL for handling high-dimensional state spaces \cite{ref3, ref21}. Architectures such as Deep Q-Networks (DQN) \cite{ref22, ref23} and Actor-Critic frameworks \cite{ref24, ref25} have been proposed for channel selection and power control.

Current DRL-based methods face challenges in massive-scale environments. Most DRL agents rely on neural networks that lack semantic understanding, resulting in slow convergence when action spaces explode. Furthermore, although RL is designed for non-stationary environments, traditional policy gradients often suffer from high variance\cite{ref26} and policy instability during multi-agent coordination. The LSA framework addresses these constraints by replacing rigid networks with an LLM engine, providing the semantic reasoning necessary for efficient decision-making.

\subsection{Scalability and Emerging Paradigms}
Recent literature has explored advanced paradigms to improve ISA scalability. Federated Learning (FL) has been integrated with MARL \cite{ref5, ref26} to facilitate distributed policy updates. To meet complex service requirements, DRL has been extended to NOMA-aided networks \cite{ref27} and mobile edge computing (MEC) scenarios \cite{ref28}. Additionally, Gaussian process-based RL \cite{ref29} and heuristically accelerated Q-learning \cite{ref30} have been investigated to improve convergence.

While these methods enhance scalability, they often suffer from a generalization gap when transitioning to unseen configurations. The emergence of LLMs provides a potential solution for zero-shot adaptation \cite{ref31, ref32, ref2, ref33}. LLMs solve the scalability bottleneck by mapping complex numerical states into low-dimensional semantic spaces, where the model's pre-trained sequence logic identifies latent coordination patterns without explicit signaling. To resolve the computational overhead of applying LLMs to high-frequency decisions, this work utilizes a GRPO mechanism to stabilize training in massive-scale dynamic networks.

\begin{table}[t]
\renewcommand{\arraystretch}{1.2}
\caption{Comparison of Spectrum Access Methodologies}
\centering
\begin{tabular}{lccc}
\toprule
\textbf{Features} & \textbf{Conventional} & \textbf{DRL-Based} & \textbf{Proposed LSA} \\ 
\midrule
Model Dependency & High & Low & Low \\
Reasoning Ability & None & None & High \\
Scalability (Nodes) & Low & Medium & High \\
Inference Latency & Low & Very Low & Moderate \\
Generalization & Poor & Moderate & Robust \\
\bottomrule
\multicolumn{4}{@{}p{\dimexpr\columnwidth-2\tabcolsep}@{}}{\small Note: LSA latency reflects macro-scheduling decision time (seconds). The framework intentionally trades microsecond response for structural adaptability and zero-shot generalization, as validated in Sec. V-A.}
\end{tabular}
\end{table}

\section{System Model and Problem Formulation}

This section establishes the formal analytical framework for the ISA problem in a large-scale wireless environment. We consider a terrestrial communication architecture where a central Base Station (BS) coordinates spectral resources among a massive set of heterogeneous nodes within a structured time-frequency grid.

\subsection{Network and Traffic Model}
We consider a wireless network consisting of a single BS and a set of registered users denoted by $\mathcal{U} = \{1, 2, \dots, U\}$. The available spectrum is partitioned into a set of orthogonal frequency channels $\mathcal{C} = \{1, 2, \dots, C\}$, each characterized by a fixed bandwidth $B_c$. The system operates over a discrete-time horizon $t \in \{1, 2, \dots, T\}$, where each slot $t$ represents a fundamental scheduling interval for spectrum coordination \cite{ref14, ref34}.

In massive-scale IoT or 6G scenarios, traffic demands are inherently sporadic and non-stationary. At each time slot $t$, a dynamic subset of users $\mathcal{K}_t \subseteq \mathcal{U}$ initiates transmission requests. The cardinality $K_t = |\mathcal{K}_t|$ fluctuates over time to reflect bursty traffic patterns, where $K_t$ varies within the dynamic range $[K_{\min}, K_{\max}] = [0.15K, 0.25K]$, under a set maximum scale parameter $K$ tailored to specific network densities. To manage this scale efficiently, the BS must execute a rapid coordination strategy to map active users to vacant channels while strictly adhering to spectral masks and primary user protection constraints \cite{ref6, ref7}.

\subsection{Transmission and Link Budget Model}
For any requesting user $u \in \mathcal{K}_t$, the signal quality at the BS depends on its spatial distribution and the underlying propagation environment. Let $(x_{\text{BS}}, y_{\text{BS}})$ and $(x_u, y_u)$ denote the Cartesian coordinates of the BS and user $u$, respectively. The normalized Euclidean distance is defined as:
\begin{equation}
d_u(t) = \frac{\sqrt{(x_u(t) - x_{\text{BS}})^2 + (y_u(t) - y_{\text{BS}})^2}}{D_{\max}},
\end{equation}
where $D_{\max}$ represents the maximum coverage radius of the cell. To accurately characterize the link budget, we employ a standard path loss model. The received signal power $P_{u,c}(t)$ for user $u$ on channel $c$ is formulated as:
\begin{equation}
\label{eq:trans_power}
P_{u,c}(t) = P_{\text{tx}} \cdot G_t \cdot G_r \cdot \left( \frac{\lambda}{4\pi d_u(t) D_{\max}} \right)^\alpha,
\end{equation}
where $P_{\text{tx}}$ is the transmit power, $G_t$ and $G_r$ are the antenna gains, $\lambda$ is the carrier wavelength, and $\alpha$ is the path loss exponent. 

Assuming an additive white Gaussian noise (AWGN) environment with noise power spectral density $N_0(t)$, the achievable transmission rate $R_{u,c}(t)$ follows the Shannon-Hartley theorem \cite{ref35}:
\begin{equation}
R_{u,c}(t) = B_c \log_2 \left( 1 + \frac{P_{u,c}(t)}{\sigma^2 + I_{u,c}(t)} \right),
\end{equation}
where $\sigma^2 = N_0(t) B_c$ denotes the total thermal noise power within the channel bandwidth $B_c$, and $I_{u,c}(t)$ represents the aggregate interference from co-channel primary users or adjacent cell leakage. In this study, we assume the BS maintains an occupancy vector $\mathbf{O}_t = [o_1(t), \dots, o_C(t)]$, where $o_c(t)=1$ indicates that channel $c$ is unavailable due to primary user activity, thus imposing a hard constraint on the ISA coordination process \cite{ref14}.

\subsection{Problem Formulation}
The objective of the spectrum coordination problem is to determine an optimal allocation matrix $\mathbf{X}_t \in \{0, 1\}^{K_t \times C}$ at each slot $t$, where the binary element $x_{u,c}(t) = 1$ if channel $c$ is assigned to user $u$, and $x_{u,c}(t) = 0$ otherwise. 

To maximize the long-term aggregate spectral utility while considering both instantaneous sum-rate and resource distribution fairness, the optimization problem is formally defined as:
\begin{align}
    \max_{\{\mathbf{X}_t\}_{t=1}^T} \quad & \sum_{t=1}^T \gamma^{t-1} \left( \sum_{u \in \mathcal{K}_t} \sum_{c=1}^C x_{u,c}(t) R_{u,c}(t) \right) \label{eq:objective} \\
    \text{s.t.} \quad & \sum_{c=1}^C x_{u,c}(t) \leq 1, \quad \forall u \in \mathcal{K}_t, \label{eq:const_user} \\
    & \sum_{u \in \mathcal{K}_t} x_{u,c}(t) + o_c(t) \leq 1, \quad \forall c \in \mathcal{C}, \label{eq:const_chan} \\
    & x_{u,c}(t) \in \{0, 1\}, \quad \forall u, c, t, 
\end{align}
where $\gamma \in [0, 1)$ is a discount factor reflecting the temporal coupling of traffic demands and spectral efficiency. 

In this formulation, \eqref{eq:const_user} represents the single-connectivity constraint, ensuring that each requesting user is served by at most one frequency channel to maintain transceiver simplicity. Constraint \eqref{eq:const_chan} enforces spectral exclusivity, which prevents co-channel interference by ensuring that a channel is not simultaneously occupied by multiple secondary users or a primary user (as indicated by $o_c(t)$). The objective \eqref{eq:objective} thus seeks a collision-free and executable scheduling policy that balances cumulative throughput with long-term network stability.

\subsection{Hardness Analysis}
To formally justify the necessity of the proposed intelligent reasoning framework, we establish the NP-hardness of the spectrum allocation problem \eqref{eq:objective} via a polynomial-time reduction from the Koopmans-Beckmann form of the Quadratic Assignment Problem (QAP). 

In ultra-dense deployments, the aggregate utility must account for co-channel interference coupling. We therefore extend the linear objective \eqref{eq:objective} to a fairness-and-interference-aware form:
\begin{equation}
\begin{split}
    \max_{\mathbf{X}} \sum_{t=1}^T \gamma^{t-1} \bigg( & \sum_{u,c} x_{u,c} R_{u,c} \\
    - & \eta \sum_{u \neq u'} \sum_{c} x_{u,c} x_{u',c} \mathcal{I}_{u,u'}(c) \bigg),
\end{split}
\end{equation}

where $\mathcal{I}_{u,u'}(c)$ models the pairwise interference penalty on channel $c$, and $\eta$ controls the coupling strength. Ignoring the temporal discount for a single-slot instance, the objective becomes quadratic in $\mathbf{X}$.

The QAP seeks a permutation $\pi \in S_n$ minimizing $\sum_{i,j} f_{ij} d_{\pi(i)\pi(j)}$. Consider a spectrum instance with $|\mathcal{K}_t| = |\mathcal{E}_t| = n$. Let the assignment matrix $x_{u,c}=1 \iff \pi(u)=c$. By mapping the user-channel coupling term $F_{uu'} \propto \mathcal{I}_{u,u'}$ and channel correlation $D_{cc'}$ to the QAP flow/distance matrices, our quadratic maximization problem is isomorphic to the QAP minimization problem under linear transformation. Since the QAP is NP-hard in the strong sense and the mapping from network parameters to $F, D$ is polynomial, the interference-aware spectrum allocation problem is NP-hard. 

Consequently, the factorial growth of the feasible set $\mathcal{X}$ renders exact branch-and-bound or exhaustive search computationally prohibitive for real-time macro-scheduling. This theoretical intractability necessitates a reasoning-based framework capable of navigating high-dimensional combinatorial landscapes within bounded latency, as detailed in Section IV.

While \eqref{eq:objective} formally defines the system-level objective, its combinatorial nature and non-differentiable constraints prevent direct gradient-based optimization. To enable data-driven policy learning within the proposed framework, we transform this objective into a differentiable scalar reward signal that explicitly couples throughput, interference mitigation, and fairness. The precise formulation of this reward function and its integration into the GRPO advantage estimation are detailed in Section~\ref{sec:reward_grpo}.

\begin{figure*}[!t]
\centering
\includegraphics[width=1.0\textwidth]{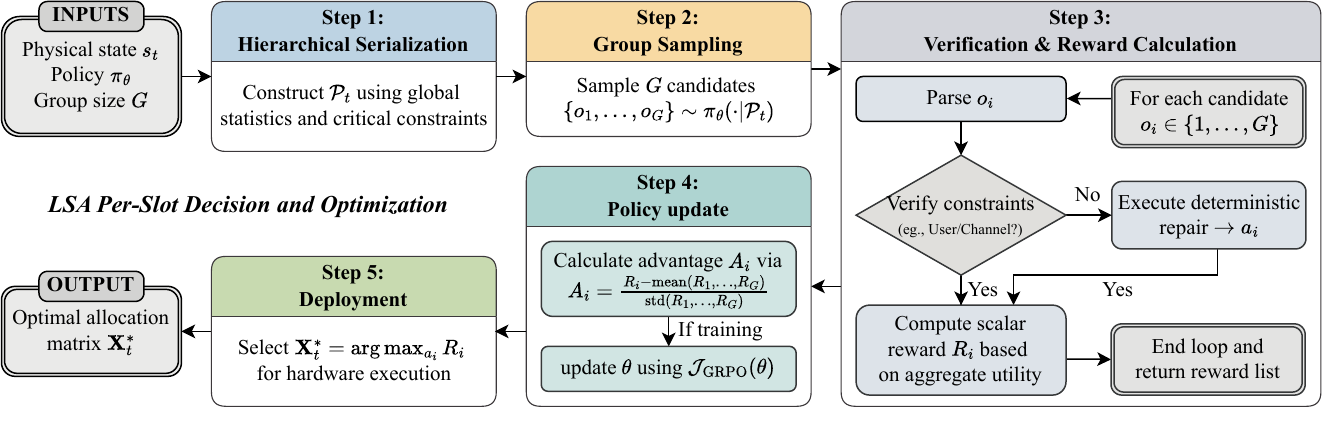}
\caption{Per-slot LSA workflow: state serialization to $\mathcal{P}_t$, group-based code generation, deterministic repair with reward evaluation, GRPO-based training, and deployment of $\mathbf{X}_t^*$.}
\label{fig:framework}
\end{figure*}

\section{Proposed LLM-Driven Spectrum Access Framework}

To address combinatorial intractability and slow convergence of DRL in large-scale deployments, we propose LSA, which reformulates spectrum allocation as structured reasoning. The framework leverages LLM in-context generalization to handle massive heterogeneous networks without explicit high-dimensional transition modeling.

\subsection{Architectural Design and Semantic Mapping}
The key challenge in large-scale spectrum access is the factorial action space, which destabilizes value-based RL. LSA replaces numerical value estimation with semantic reasoning and consists of three components:

First, a \textit{Semantic Feature Mapping} converts physical states into structured language representations for coordination pattern extraction.  
Second, \textit{Group-based Policy Alignment} adopts GRPO to stabilize updates in discrete action spaces.  
Third, a \textit{Deterministic Repair Pipeline} enforces constraints in \eqref{eq:const_user} and \eqref{eq:const_chan}.

\subsection{Hierarchical State Serialization and Prompt Engineering}
To bridge numerical states and linguistic reasoning, we propose \textit{Hierarchical State Serialization}, which converts $s_t$ into a compact prompt $\mathcal{P}_t$.

The prompt contains three semantic blocks:

\begin{itemize}
\item \textbf{Network Statistical Context}: aggregates macroscopic features such as noise floor $N_0(t)$, active users $K_t$, and summary statistics (mean, max, std) of $\{P_{u,c}\}$ and $\{B_c\}$. This reduces token cost while preserving coordination signals.

\item \textbf{Operational Constraints and Guidance}: embeds Shannon-Hartley law and unit conversion rules, explicitly enforcing constraints \eqref{eq:const_user} and \eqref{eq:const_chan} as structured directives for collision-free assignment.

\item \textbf{Output Schema Specification}: restricts output to a Python dictionary \texttt{action = \{user\_id: channel\_id\}}, ensuring executable and structured generation using standard scientific libraries.
\end{itemize}

Overall, the prompt is designed to maximize reasoning efficiency under a 2048-token limit, even in ultra-dense scenarios.

\textit{Remark 2 (Signaling Overhead)}: Aggregating global statistics introduces communication overhead, but this cost is amortized since LSA operates as a macro-scheduler over longer time horizons. Compared to continuous coordination in multi-agent DRL, the overhead remains significantly lower.

\subsection{Dual-Paradigm Training Dynamics}

We compare two training strategies: Supervised Fine-Tuning (SFT) + GRPO and direct GRPO on code models.

\textbf{1) Text-Based Pipeline (SFT + GRPO).}  
SFT stabilizes early training using expert dataset $\mathcal{D}$:

\begin{equation}
\mathcal{J}_{\text{SFT}}(\theta) = \mathbb{E}_{(P, o^*) \sim \mathcal{D}} \left[ \sum_{\tau=1}^{|o^*|} \log \pi_\theta(o_\tau^* \mid o_{<\tau}^*, P) \right].
\end{equation}

It provides syntactic priors and reduces invalid outputs. GRPO then refines the policy with KL anchoring to $\pi_{\text{SFT}}$, ensuring bounded exploration.

\textbf{2) Code-Specialized Pipeline (Direct GRPO).}  
The code backbone skips SFT due to strong pretraining on algorithmic corpora. GRPO is applied directly with group-relative advantages:

\begin{equation}
A_i = \frac{R(o_i) - \mu_R}{\sigma_R + \epsilon}.
\label{eq:advantage}
\end{equation}

The final objective:

\begin{equation}
\begin{split}
\mathcal{J}_{\text{GRPO}}(\theta) = \mathbb{E}_{P_t, \{o_i\}} \Bigg[ & \frac{1}{G}\sum_{i=1}^G \min\Big( r_i(\theta) A_i, \\
& \text{clip}(r_i(\theta), 1-\epsilon_c, 1+\epsilon_c)A_i \Big) \\
& - \beta \mathbb{D}_{\text{KL}}(\pi_\theta || \pi_{\text{ref}}) \Bigg],
\end{split}
\label{eq:grpo_final}
\end{equation}

where $r_i(\theta)$ is the importance ratio and $\epsilon_c$, $\beta$ control update stability.

This approach improves stability on sparse execution rewards, as GRPO filters variance without requiring a critic.

Overall, performance differences arise from representation bias: the Coder backbone directly operates in a structured semantic manifold aligned with execution logic, while SFT introduces an intermediate imitation bottleneck and distribution shift.

\subsection{Execution-Aware Composite Reward Formulation}
\label{sec:reward_grpo}

The reward signal bridges token-level generation and physical network utility. To mitigate gradient scale imbalance, we define a normalized tripartite reward:
\begin{equation}
R(o) = w_1 \hat{R}_{\text{struct}}(o) + w_2 \hat{R}_{\text{perf}}(o) + w_3 \hat{R}_{\text{depth}}(o),
\end{equation}
where $w_1, w_2, w_3$ control objective weights and each $\hat{R}_k(\cdot)$ is rescaled to $[-10, 10]$ to ensure magnitude consistency. This normalization prevents structural signals from dominating gradients due to discrete execution feedback.

During training, the deterministic repair pipeline $\mathcal{R}(\cdot)$ is disabled so that rewards are computed on raw outputs $o$. Constraint-violating generations receive strong structural penalties ($\hat{R}_{\text{struct}} \to 0$), enforcing internalization of syntactic and logical priors. In contrast, $\mathcal{R}(\cdot)$ is only activated during inference to guarantee feasibility.

\noindent \textbf{1) Hierarchical Structural Compliance ($\hat{R}_{\text{struct}}$):}
Structural compliance is defined as a progressive validation signal rather than a binary filter, aggregating: \textit{syntactic executability}, \textit{constraint alignment} with active system states, and \textit{context-aware reasoning}. The reward is capped at 1.2 and acts as a soft curriculum that penalizes hallucinated or shallow outputs while guiding learning toward constraint-aware generation.

Compared to a binary filter, this design avoids excessive rejection of early-stage samples (i.e., $\hat{R}_{\text{struct}} \to 0$ for most outputs), which would otherwise induce gradient sparsity and destabilize training. Instead, progressive scoring improves exploration while maintaining constraint awareness.

\noindent \textbf{2) Relative Performance Gain ($\hat{R}_{\text{perf}}$):} To ensure density-invariant optimization, throughput is normalized against a dynamically generated random valid assignment $\mathbf{a}_{\text{rand}}$:
\begin{equation}
\hat{R}_{\text{perf}}(o) = \left( \frac{\mathcal{T}(\mathbf{a}_{\text{LLM}})}{\mathcal{T}(\mathbf{a}_{\text{rand}})} - 1 \right) \cdot \psi_{\text{penalty}}(\mathbf{a}_{\text{LLM}}),
\end{equation}
where $\mathcal{T}(\cdot)$ computes aggregate Shannon capacity. Crucially, this stochastic baseline expands the reward dispersion within each sampled group. In high-dimensional combinatorial landscapes, naive sampling often yields heavily clustered rewards, causing the intra-group reward variance $\sigma_R \to 0$ in \eqref{eq:advantage}, which leads to gradient collapse. By anchoring evaluations to $\mathbf{a}_{\text{rand}}$, the system preserves a robust variance $\sigma_R$, enabling continuous policy exploration.

\noindent \textbf{3) Reasoning Depth Regularization ($\hat{R}_{\text{depth}}$):} A length-aware incentive term is introduced to enforce sufficient reasoning depth while preventing premature code truncation:
\begin{equation}
\hat{R}_{\text{depth}}(o) = - \max\left(0, \frac{L_{\text{thr}} - L_{\text{out}}}{L_{\text{thr}}} \cdot \kappa \right),
\end{equation}
where $L_{\text{out}}$ is the token length of the generated response, $L_{\text{thr}} = L_{\max}/2 = 512$ defines the minimum reasoning depth threshold, and $\kappa=5$ sets the penalty magnitude. 
This formulation establishes a ``soft lower bound'' on output length: responses shorter than the threshold incur a linearly scaled penalty (capped at $-5$), which saturates to zero once $L_{\text{out}} \geq L_{\text{thr}}$. Unlike traditional length penalties that indiscriminately suppress verbosity, this design specifically penalizes shallow, truncated generations, compelling the policy to produce step-by-step constraint derivations and coherent algorithmic chains. The deterministic repair pipeline subsequently guarantees physical feasibility before deployment.

\textit{Remark 3 (Hyperparameter Sensitivity):} While extreme operational deviations can induce reward-hacking—such as over-weighting performance metrics ($w_2=2$) at the cost of hard syntactic alignment—minor perturbations around the baseline configurations ($w_1=1, w_2=1, w_3=1$) show remarkable convergence stability. This robustness stems from the structural decoupling of the curriculum reward combined with the $[-10, 10]$ dynamic range boundaries, which effectively neutralizes single-metric gradient dominance.

\begin{algorithm}[t]
\caption{LSA Decision and Policy Alignment}
\begin{algorithmic}[1]
\State \textbf{Input:} Physical state $s_t$, policy $\pi_\theta$, group size $G$
\State \textbf{Output:} Optimal allocation matrix $\mathbf{X}_t^*$
\State \textbf{1. Hierarchical Serialization:} 
\State \quad Construct $\mathcal{P}_t$ using global statistics and critical constraints.
\State \textbf{2. Group Sampling:} 
\State \quad Sample $G$ candidates $\{o_1, \dots, o_G\} \sim \pi_\theta(\cdot | \mathcal{P}_t)$.
\State \textbf{3. Verification and Reward Calculation:}
\For{each candidate $o_i \in \{1, \dots, G\}$}
    \State Parse $o_i$ and verify constraints \eqref{eq:const_user}--\eqref{eq:const_chan}.
    \If{Inference Phase}
        \State Execute deterministic repair $\mathcal{R}(o_i) \rightarrow a_i$.
    \Else
        \State Set $a_i \leftarrow o_i$ (repair disabled for training gradient purity).
    \EndIf
    \State Compute scalar reward $R_i$ based on aggregate utility.
\EndFor
\State \textbf{4. Policy Update:} 
\State \quad Calculate advantage $A_i$ via \eqref{eq:advantage} and update $\theta$ using \eqref{eq:grpo_final}.
\State \textbf{5. Deployment:} 
\State \quad Select $\mathbf{X}_t^* = \arg\max_{a_i} R_i$ for hardware execution.
\end{algorithmic}
\end{algorithm}

\subsection{Feasibility Verification and Deterministic Repair} \label{sec:repair}

To ensure robustness in large-scale deployments where LLM outputs may violate constraints due to stochastic token generation, we introduce a deterministic repair pipeline. For a serialized state $\mathcal{P}_t$, the raw LLM output is denoted as $\mathbf{m}_t$.

We represent the allocation matrix $\mathbf{X}_t$ as a mapping $\mathbf{a}_t: \mathcal{K}_t \rightarrow \mathcal{E}_t$, where $\mathbf{a}_t(u)=c \Leftrightarrow x_{u,c}(t)=1$. The feasible channel set is defined as $\mathcal{E}_t = \{c \in \mathcal{C} \mid o_c(t)=0\}$. We assume $|\mathcal{E}_t| \ge |\mathcal{K}_t|$.

\begin{theorem}[Feasibility and Local Optimality of Repair]
Under $|\mathcal{E}_t| \ge |\mathcal{K}_t|$, the repair operator $\mathcal{R}(s_t, \mathbf{m}_t)$ produces a feasible mapping $\mathbf{a}_t$ such that:
\begin{equation}
\mathbf{a}_t(u) \in \mathcal{E}_t, \forall u \in \mathcal{K}_t, \quad \mathbf{a}_t(u) \neq \mathbf{a}_t(u'), \forall u \neq u'.
\end{equation}
It preserves valid LLM suggestions and maximizes throughput within conflicting subsets.
\end{theorem}

\begin{IEEEproof}
The repair operator consists of three steps. First, invalid outputs are filtered and valid pairs $(u,c)$ are extracted from $\mathbf{m}_t$. Second, for each conflicting channel $c$, where $\mathcal{U}_c \subseteq \mathcal{K}_t$, we apply:
\begin{equation}
u^* = \arg\max_{u \in \mathcal{U}_c} R_{u,c}(t).
\label{eq:winner_take_all}
\end{equation}
This yields a locally optimal assignment for each contention group.

Third, remaining users are greedily assigned to unused channels in $\mathcal{E}_t \setminus \{\mathbf{a}_t(u^*)\}$, ensuring a feasible complete mapping in $\mathcal{X}$. Under standard submodular utility assumptions, this greedy resolution achieves a constant-factor approximation to bipartite matching, bridging generative outputs and hard constraints.
\end{IEEEproof}

\textit{Remark 1 (Theoretical Guarantees and Action Mapping):}
The LLM output is constrained to a structured action vocabulary and parsed deterministically into a mapping dictionary. Theorem 1 guarantees that stochastic token sampling cannot violate spectral exclusivity by construction. The greedy conflict resolution admits a constant-factor approximation bound under submodular utilities.

The overall per-slot complexity of Algorithm 1 is: $\mathcal{O}(G \cdot L^2 d + K_t \log K_t)$, where $G$ is group size, $L$ is sequence length, $d$ is embedding dimension, and $K_t$ is active users. Due to hierarchical serialization, $L$ is bounded within 2048 tokens, scaling sublinearly with system size and improving over $\mathcal{O}(N^3)$ classical solvers.

For deployment efficiency, we adopt a static–dynamic KV-cache strategy: dynamic historical allocations are flushed at each slot $t$, while static constraint prompts are cached across slots to reduce redundant computation.

Finally, under imperfect CSI, estimation noise may affect the greedy selection in Eq.~\eqref{eq:winner_take_all}, potentially leading to suboptimal local assignments. However, the deterministic repair pipeline guarantees feasibility and ensures only bounded throughput degradation without violating collision-free constraints.

\begin{table}[!t]
\renewcommand{\arraystretch}{1.3}
\caption{System Simulation Parameters and Model Configuration}
\label{tab:config}
\centering
\begin{tabular}{ll}
\toprule
\textbf{Category} & \textbf{Parameter and Value} \\
\midrule
\multirow{3}{*}{Network Topology} & Cell Radius: $500$ m \\
& User Distribution: Spatial Poisson Process \\
& Noise Density: $-112$ dBm/Hz \\
\hline
\multirow{3}{*}{Link Budget} & Transmit Power $P_{\text{tx}}$: $23$ dBm \\
& Path Loss Exponent $\alpha$: $3.5$ \\
& Channel Bandwidth $B_c$: $5 \sim 20$ MHz \\
\hline
\multirow{3}{*}{Reward Formulation} & Structural Weight $\omega_1$: $1.0$ \\
& Performance Weight $\omega_2$: $1.0$ \\
& Depth Regularization $\omega_3$: $1.0$ \\
\hline
\multirow{4}{*}{Model Configuration} & Text Backbone: Qwen2.5-7B-Instruct \\
& Coder Backbone: Qwen2.5-Coder-7B-Instruct \\
& GRPO Group Size $G$: $8$ \\
& GRPO KL Coefficient $\beta$: $0.25$ \\
\bottomrule
\end{tabular}
\end{table}

\section{Experimental Setup}

\subsection{Network Scenario and Simulation Environment}
The simulation models a single-cell terrestrial network where nodes follow a homogeneous Poisson point process within a 500 m radius. Spectrum is divided into orthogonal channels with bandwidths uniformly sampled from $[5,20]$ MHz. We adopt a log-distance path loss model ($\alpha=3.5$) and AWGN noise at $-112$ dBm/Hz. Traffic is bursty and non-stationary, with active user sets $K_t \subset U$ varying per scheduling slot. Spectral exclusivity and primary-user protection are enforced as hard constraints during reward computation and post-inference verification.

To enable semi-physical deployment, the 7B backbone can run on edge devices such as NVIDIA Jetson modules. With INT4/INT8 quantization, structured pruning, and KV-cache offloading, memory usage is reduced below 6 GB, enabling low-power macro-scheduling at gateway scale.

\subsection{Baseline Methodologies and Compute-Equivalent Protocol}
We compare against four categories:
\begin{itemize}
\item \textbf{Combinatorial Optimization:} Exhaustive Enumeration provides small-scale upper bounds. The KM algorithm \cite{ref_km} is used as a standard baseline, while Grouped-KM partitions graphs to handle $O(N^3)$ complexity in large-scale settings.
\item \textbf{Meta-Heuristics:} Differential Evolution (DE) \cite{ref_de} uses population size 50 and crossover rate 0.9 as a classical iterative optimizer without learning priors.
\item \textbf{Reinforcement Learning:} DQN \cite{ref_dqn} and PPO \cite{ref_ppo} with 128-d hidden layers are used, with output heads dynamically resized to match action space cardinality, testing robustness under combinatorial explosion.
\item \textbf{Stochastic Baseline:} Uniform random allocation provides a lower-bound reference.
\end{itemize}

All learning-based methods are evaluated under a compute-equivalent protocol with matched GPU FLOPs and runtime budgets to isolate algorithmic efficiency. Results are averaged over Monte Carlo runs, with additional analysis of stability and convergence behavior.

\subsection{Evaluation Metrics}
We evaluate (1) \textit{Aggregate Throughput}, defined as mean Shannon capacity per slot over active users; (2) \textit{Valid Action Rate}, measuring constraint-satisfying allocations before repair; and (3) \textit{Inference Latency}, recorded as wall-clock decision time. Zero-shot generalization is tested on unseen topologies and densities without any parameter or prompt updates.

\section{Experimental Results and Analysis}

\subsection{Small-to-Medium Scale Benchmarking}

\begin{table*}[!t]
\renewcommand{\arraystretch}{1.2}
\caption{Aggregate Throughput Performance Across Varying Network Densities}
\setlength{\tabcolsep}{10pt}
\label{tab:performance}
\centering
\begin{tabular}{lcccccccc}
\toprule
\textbf{Algorithm} & \textbf{(5,5,1)} & \textbf{(7,10,2)} & \textbf{(10,15,3)} & \textbf{(15,20,4)} & \textbf{(5,5,2)} & \textbf{(8,10,4)} & \textbf{(12,15,6)} & \textbf{(15,20,8)} \\
\midrule
LLM (Coder) & 120.98 & 248.36 & 405.96 & \underline{532.81} & 239.23 & \underline{500.25} & \textbf{785.02} & \underline{1008.33} \\
Random Allocation & 110.63 & 221.66 & 355.97 & 468.10 & 224.93 & 464.85 & 724.92 & 940.30 \\
Exhaustive Enumeration & \underline{126.77} & \textbf{281.62} & \textbf{451.34} & \textbf{582.43} & \underline{244.15} & \textbf{538.26} & N/A$^*$ & N/A$^*$ \\
KM Algorithm & 118.42 & 249.75 & \underline{419.89} & 528.71 & 231.61 & 489.14 & \underline{746.99} & \textbf{1044.25} \\
Differential Evolution & 92.58 & 193.46 & 320.70 & 406.48 & 184.39 & 383.52 & 585.50 & 789.58 \\
DQN & \textbf{127.29} & \underline{250.26} & 402.83 & 499.56 & \textbf{249.93} & 481.40 & N/A$^*$ & N/A$^*$ \\
PPO & 108.87 & 239.58 & 380.04 & 479.58 & 234.12 & 467.53 & N/A$^*$ & N/A$^*$ \\
\bottomrule
\multicolumn{9}{l}{\small $^*$ N/A indicates memory exhaustion due to action space explosion.}
\end{tabular}
\vspace{0.2em}
\hbox{\small \textit{Note:} The network configuration tuple $(U, C, K)$ explicitly denotes (Total Users, Total Channels, Active Users) respectively.}
\end{table*}

\begin{table}[!t]
\renewcommand{\arraystretch}{1.2}
\caption{Inference Latency and Computational Scalability Profile (Unit: Seconds)}
\label{tab:runtime}
\centering
\begin{tabular}{lcccc}
\toprule
\textbf{Algorithm} & \textbf{(5,5,1)} & \textbf{(10,15,3)} & \textbf{(15,20,4)} & \textbf{(15,20,8)} \\
\midrule
LLM (Coder) & 7.8282 & 7.2637 & 7.5774 & 7.2067 \\
Random Allocation & 0.0000 & 0.0000 & 0.0001 & 0.0000 \\
Exhaustive Enumeration & 0.0001 & 0.0540 & 2.7850 & N/A$^*$ \\
KM Algorithm & 0.0002 & 0.0017 & 0.0047 & 0.0155 \\
DQN & 0.0012 & 0.0108 & 0.0180 & N/A$^*$ \\
PPO & 0.0007 & 0.0008 & 0.0015 & N/A$^*$ \\
\bottomrule
\multicolumn{5}{@{}p{\dimexpr\columnwidth-2\tabcolsep}@{}}{\small $^*$ N/A indicates memory/time exhaustion. LLM latency reflects second-level macro-scheduling cycles.}
\end{tabular}
\end{table}

\begin{table}[!t]
\renewcommand{\arraystretch}{1.2}
\caption{Ultra-Large-Scale Performance Scaling and Architectural Comparison}
\label{tab:large_scale}
\centering
\begin{tabular}{lccc}
\toprule
\textbf{Algorithm} & \textbf{Scale I} & \textbf{Scale II} & \textbf{Scale III} \\
\midrule
LLM (Text, SFT+GRPO) & 22893 & N/A$^\dagger$ & N/A$^\dagger$ \\
LLM (Coder, Pure GRPO) & 26545 & \textbf{41891} & \textbf{55099} \\
Grouped KM & \textbf{27703} & 38540 & 53096 \\
Random Allocation & 20713 & 32702 & 43159 \\
\bottomrule
\multicolumn{4}{@{}p{\dimexpr\columnwidth-2\tabcolsep}@{}}{\small $^\dagger$ N/A indicates completion length exceeding context window or training divergence.}
\end{tabular}
\end{table}

\begin{table}[!t]
\renewcommand{\arraystretch}{1.2}
\caption{Impact of Context Window Size on Throughput and Inference Overhead}
\label{tab:serialization_abl}
\centering
\begin{tabular}{lcc}
\toprule
\textbf{Context Constraint} & \textbf{Throughput} & \textbf{Latency (s)} \\
\midrule
1024 tokens & 36254 & \textbf{4.6313} \\
2048 tokens (Baseline) & \textbf{41891} & 7.4582 \\
4096 tokens & 37158 & 16.3234 \\
\bottomrule
\end{tabular}
\end{table}

\begin{table*}[!t]
\renewcommand{\arraystretch}{1.2}
\caption{Topology-Agnostic Adaptation Capability and Retraining Overhead Across Network Scales}
\label{tab:generalization}
\centering
\begin{tabular}{lccc}
\toprule
\textbf{Algorithm} & \begin{tabular}[c]{@{}c@{}}Parameter\\Retraining?\end{tabular} & \begin{tabular}[c]{@{}c@{}}Adaptation\\Mechanism\end{tabular} & \begin{tabular}[c]{@{}c@{}}Topology Shift\\Resilience\end{tabular} \\
\midrule
LLM (Coder/Text) & No & Prompt Context Injection (ICL) & Robust \\
DQN / PPO & Yes & Architecture Reshape \& Full Retraining & Fragile \\
KM / Grouped KM & No & Static Matrix Recomputation & Rigid (Partitioning) \\
Differential Evolution & No & Heuristic Parameter Retuning & Moderate \\
\bottomrule
\end{tabular}
\end{table*}

\begin{figure*}[!t]
\centering
\includegraphics[width=1.0\textwidth]{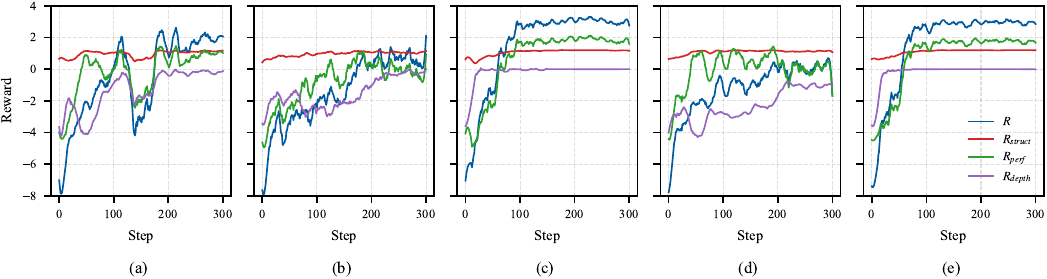}
\caption{GRPO training dynamics and reward component evolution under varying hyperparameters. Subfigures depict (a) Baseline ($G=8, \omega_1=1, \omega_2=1, \omega_3=1$), (b) Small group ($G=4, \omega_1=1, \omega_2=1, \omega_3=1$), (c) Structure-prioritized ($G=8, \omega_1=2, \omega_2=1, \omega_3=1$), (d) Performance-overweighted ($G=8, \omega_1=1, \omega_2=2, \omega_3=1$), and (e) Depth-regularized ($G=8, \omega_1=1, \omega_2=1, \omega_3=2$). Each subplot tracks total reward alongside $R_{\text{struct}}$, $R_{\text{perf}}$, and $R_{\text{depth}}$ components over 300 training steps.}
\label{fig:reward_trend_comparison}
\end{figure*}

\begin{table}[!t]
\renewcommand{\arraystretch}{1.2}
\caption{Training Paradigm Efficacy and Reward Component Necessity}
\label{tab:paradigm_abl}
\centering
\begin{tabular}{lcc}
\toprule
\textbf{Configuration} & \textbf{Throughput} & \textbf{Valid Rate$^\ddagger$} \\
\midrule
Coder + Pure GRPO & \textbf{26545} & \textbf{100\%} \\
Text + SFT + GRPO & 22893 & 78\% \\
Coder + Pure GRPO (w/o $\omega_1$) & N/A & 0\% \\
Coder + Pure GRPO (w/o $\omega_3$) & N/A & 0\% \\
\bottomrule
\multicolumn{3}{@{}p{\dimexpr\columnwidth-2\tabcolsep}@{}}{\small $^\ddagger$ Valid Rate reports the proportion of syntactically and constraint-compliant \textit{raw LLM outputs} prior to any repair. The repair pipeline is disabled during ablation to isolate the intrinsic policy alignment capability. N/A indicates training divergence.}
\end{tabular}
\end{table}

Table~\ref{tab:performance} summarizes aggregate throughput across eight combinatorial configurations with $(U,C) \in [5,15]$. In low-density settings, the Coder-based LSA achieves $95.5\%$ of the Exhaustive Enumeration upper bound at $(10,15,3)$, indicating near-optimal coordination. As scale increases to $(15,20,8)$, it maintains $1008.33$ throughput, outperforming Differential Evolution ($789.58$) and all DRL baselines.

DQN and PPO degrade sharply beyond medium scale due to action-space explosion and reward sparsity, often failing to produce valid mappings. The KM method remains competitive but requires full recomputation under objective changes, whereas LSA adapts via semantic prompt restructuring without retraining.

The performance gap is driven by representation differences: DRL compresses high-dimensional state information into fixed latent vectors, losing combinatorial structure as action space grows factorially. In contrast, LSA preserves dependencies via hierarchical serialization and multi-head attention, enabling dynamic allocation of computation to key constraint interactions and providing a stronger inductive bias than MLP-based agents.

Table~\ref{tab:runtime} further reports inference latency. While classical solvers operate at millisecond scale for MAC-layer scheduling, LSA requires $7.2$--$7.8$s per slot, positioning it as a quasi-static macro-scheduler for gateway orchestration, network slicing, or digital-twin planning. The deterministic repair pipeline guarantees $100\%$ feasibility, trading latency for structural adaptability and zero-shot generalization across dynamic topologies, which is difficult for rigid numerical solvers.

\subsection{Ultra-Large-Scale Scaling Laws and Paradigm Validation}

To evaluate asymptotic behavior, we define three ultra-dense topologies: \textbf{Scale I} $(U,C,K)=(1000,1000,200)$, \textbf{Scale II} $(1500,1500,300)$, and \textbf{Scale III} $(2000,2000,400)$. At Scale III, the combinatorial action space expands to $C^{K_t} = 2000^{400} \approx 2.58 \times 10^{1320}$ states, highlighting the extreme intractability of exact matrix-based and classical RL formulations and motivating our low-dimensional hierarchical semantic representation.

Throughput results in Table~\ref{tab:large_scale} reveal a clear scaling divergence. The Text-based SFT+GRPO model fails beyond $1000$ nodes due to context overflow, producing N/A results. In contrast, the Coder-based Pure GRPO model exhibits strong scaling behavior, reaching $41891$ at Scale II and $55099$ at Scale III.

All Grouped-KM baselines operate under the same inference time budget as the LLM pipeline ($\sim 7.5$s), within which they iteratively refine partitioning and matching. Despite this advantage, Coder consistently outperforms Grouped-KM by $8.7\%$ at Scale II and $3.8\%$ at Scale III. This under equal-time comparison demonstrates that execution-driven policy optimization captures global coordination more effectively than partitioned classical heuristics.

Overall, the hierarchical serialization enables single-pass reasoning over aggregated constraints, mitigating combinatorial explosion, while Grouped-KM remains limited by boundary artifacts and subgraph synchronization overhead even under extended optimization time. This crossover validates code-specialized generative architectures as a more efficient paradigm for large-scale spectral allocation.

\subsection{Context Window Optimization and Latency Trade-offs}

The hierarchical state serialization mechanism directly governs the token budget and reasoning fidelity. Table~\ref{tab:serialization_abl} quantifies the impact of context window constraints on \textbf{Scale II} performance. A $1024$-token limit forces aggressive state aggregation, discarding fine-grained channel correlations and resulting in a $13.5\%$ throughput drop. Conversely, expanding to $4096$ tokens increases latency by $118\%$ due to quadratic self-attention scaling, while throughput declines by $11.3\%$ as the model attends to redundant context noise and KV-cache overhead dominates the inference pipeline. The $2048$-token configuration establishes the Pareto optimum, balancing high-dimensional state representation with real-time macro-scheduling feasibility. This empirical validation aligns with the information bottleneck principle, demonstrating that optimal context compression maximizes the mutual information between serialized prompts and allocation decisions.

\subsection{Training Dynamics, Ablation, and Topology-Agnostic Adaptation}

The effectiveness of the tripartite reward design and GRPO training is evaluated via learning curves and ablation studies. Figure~\ref{fig:reward_trend_comparison} shows the reward evolution over 300 steps. Configurations (c) and (e), which emphasize structural compliance ($\omega_1=2$) and reasoning depth ($\omega_3=2$), converge within $\sim 100$ steps, indicating that explicit formatting and regularization stabilize the optimization process. The baseline configuration (a) shows a transient dip around step 115 due to exploration–exploitation transition, but converges by step 250, validating the fixed 300-step training budget as sufficient for stable convergence across all settings.

Configuration (b) ($G=4$) exhibits degraded performance due to insufficient group diversity, leading to slower convergence. Configuration (d) ($\omega_2=2$) converges to a bottleneck plateau ($17762$), suggesting that over-weighting throughput induces reward-hacking behavior.

Table~\ref{tab:paradigm_abl} (Scale I) further confirms the necessity of each reward component. Removing $R_{\text{struct}}$ or $R_{\text{depth}}$ leads to training collapse with $0\%$ valid actions, highlighting the importance of syntactic feasibility and output regularization. The Coder-based paradigm achieves a $100\%$ valid rate, compared to $78\%$ for the Text baseline, demonstrating the advantage of code-based reasoning for structured allocation.

Table~\ref{tab:generalization} evaluates topology-agnostic generalization. Unlike DRL methods that require retraining under distribution shifts, the proposed LSA framework leverages pre-trained code priors and hierarchical serialization to map unseen network states into a structured semantic space, enabling in-context adaptation without parameter updates. This replaces gradient-based adaptation with reasoning-driven generalization, improving robustness across varying network scales.

\textit{Boundary of Zero-Shot Generalization:} Under highly clustered spatial distributions such as a Mat{\'e}rn process, local interference density increases significantly, which may reduce the precision of initial allocations due to compressed global state representations. However, the deterministic repair pipeline (Theorem 1) guarantees feasibility by resolving constraint violations via greedy reassignment, ensuring bounded performance degradation without system failure.

\section{Conclusion}
This paper presents a generative coordination paradigm for DSA to address combinatorial intractability in large-scale 6G IoT networks. By reformulating spectrum allocation as a structured reasoning problem, the proposed LSA framework decouples decision efficiency from the exponential growth of the action space, overcoming limitations of value-based RL and classical bipartite matching methods. Through integration of GRPO and a deterministic repair pipeline, LSA achieves a balanced trade-off between combinatorial reasoning and physical feasibility. Extensive evaluations across varying network densities show that LSA maintains strong spectral efficiency and structural stability in ultra-dense settings up to 2000 nodes and channels. In contrast, traditional DRL approaches collapse due to action-space explosion (N/A at $\sim$20 nodes), while classical solvers suffer synchronization bottlenecks under comparable compute budgets. Overall, the results demonstrate that LLM-driven structured reasoning provides a scalable alternative to numerical optimization for wireless resource management, enabling interpretable and robust allocation in future hyper-connected networks. A key open direction is the development of cross-scale lightweight architectures to adapt large generative priors to extreme edge-compute environments.

\end{document}